\documentclass[11pt]{article}
\pdfoutput=1

\usepackage{url}
\usepackage{mathtools}
\usepackage[utf8]{inputenc}
\usepackage{color}
\usepackage{xspace}
\usepackage{lineno}
\usepackage{graphicx}
\graphicspath{{figures/}}
\usepackage{amsmath,amssymb,amsthm}
\usepackage[top=1in, bottom=1in, left=1in, right=1in]{geometry}
\usepackage{authblk}

\usepackage{algorithm}
\usepackage{algorithmicx}

\algblockdefx[Foreach]{Foreach}{EndForeach}[1]{\textbf{foreach} #1 \textbf{do}}{\textbf{end foreach}}
\algnewcommand{\algorithmicinput}{\textbf{Input:}}
\algnewcommand{\algorithmicoutput}{\textbf{Output:}}
\algnewcommand\Input{\item[\algorithmicinput]\quad}%
\algnewcommand\Output{\item[\algorithmicoutput]}%

\newcommand{\NP}{$\mathrm{NP}$}
\newcommand{\cost}{\textsc{cost}\xspace}
\newcommand{\OPT}{\textsc{Opt}\xspace}
\newcommand{\BoC}{\textsc{BoC}\xspace}
\newcommand{\noC}{\mathrm{noC}}
\newcommand{\bigO}{\mathcal{O}}
\DeclareMathOperator*{\argmin}{arg\,min}

\DeclareMathOperator*{\dist}{dist}

\newtheorem{fact}{Fact}
\newtheorem{myremark}{Remark}
\newtheorem{theorem}{Theorem}

\newtheorem{definition}{Definition}

\newtheorem{example}{Example}

\begin{document}

\title{Truthful Mechanisms for Delivery with Agents}

\author[1]{Andreas Bärtschi}
\author[1]{Daniel Graf}
\author[1]{Paolo Penna}
\affil[1]{
	ETH Zürich, Department of Computer Science, Switzerland\\
	\texttt{$\left\{\right.$andreas.baertschi,\,daniel.graf,\,paolo.penna$\left.\right\}$@inf.ethz.ch}
}

\maketitle
\begin{abstract}
	We study the game-theoretic task of selecting mobile agents to deliver multiple items on a network. 
	An instance is given by $m$ packages (physical objects) which have to be transported between specified source--target pairs in an 
	undirected graph, 
	and $k$ mobile heterogeneous agents, each being able to transport one package at a time. 
	Following a recent model~\cite{STACS17}, 
	each agent $i$ has a different rate of energy consumption per unit distance traveled, i.e., its \emph{weight}. 
	We are interested in optimizing or approximating the \emph{total energy consumption} over all selected agents.

	Unlike previous research, we assume the weights to be private values known only to the respective agents. 
	We present three different mechanisms which select, route and pay the agents in a truthful way that guarantees voluntary participation of the agents, 
	while approximating the optimum energy consumption by a constant factor.
	To this end, we analyze a previous structural result and an approximation algorithm given in~\cite{STACS17}.
	Finally, we show that for some instances in the case of a single package, the sum of the payments can be bounded in terms of the optimum. 
\end{abstract}

\section{Introduction}
\label{sec:introduction}

We study the \emph{delivery} of physical objects (henceforth called \emph{packages}) by mobile agents. 
Regardless of whether large volumes are transported by motor lorries or cargo airplanes, or whether autonomous drones deliver your groceries, the fuel or battery consumption of the transporting agent is a major cost factor.
The main concern for algorithm-design thus focuses on an energy-efficient operation of these \emph{agents}. 
The primary energy expense is defined by the movements of the agents; in this paper we consider the energy consumption to be proportional to the distance traveled by an agent. 
We assume the agents to be heterogeneous in the sense of the agents having different rates of energy consumption.

Recent progress in minimizing the \emph{total energy consumption} in \emph{delivery problems} 
has been made assuming that the agents belong to the same entity which wants to deliver the packages~\cite{STACS17}. 
It has been shown that there is a polynomial-time constant-factor approximation algorithm with the ratio
depending on the energy consumption rates.

In the following, however, we assume the agents to be \emph{independent selfish agents} which make a bid to transport packages.
Think for instance of a cargo company that hires subcon\-tractors (like independent lorry drivers or individual couriers) to outsource the delivery of packages to these agents.
In this scenario, the rate of energy consumption is a private value known only to the respective agent.  
Our goal is hence to design a mechanism for the cargo company to select agents, plan the agents' routes and reimburse them in a way such that:
\begin{enumerate}
	\item	the energy spent overall is as close to the optimum as possible,
	\item	each agent announces its true rate of energy consumption in its bid, and 
	\item	each agent is reimbursed at least the cost of energy it needs to deliver all of its packages. 
\end{enumerate}

\begin{figure*}[t!]
	\includegraphics[width=\linewidth]{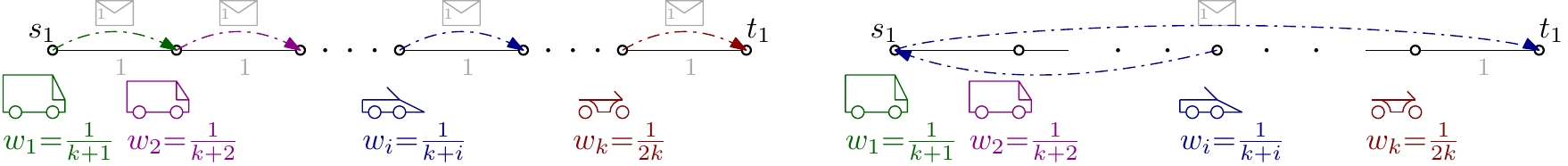}
	\caption{Delivery of a single package on a path of length $k$; consecutive agents of weight $w_i = \tfrac{1}{k+i}$.\newline
	(left) Optimal solution using all agents; cost 
	$\smash{= \sum_{i=1}^k w_i\cdot 1 = \sum_{i=1}^k \tfrac{1}{k+i} = \mathcal{H}_{2k} - \mathcal{H}_k \approx \ln 2\ (k \rightarrow \infty)}$.\newline
	(right) Any solution using a single agent $i$ has cost 
	$w_i \cdot \left( (i-1) + k \right) = (k+i-1)/(k+i) \approx 1\ (k \rightarrow \infty)$.\label{fig:delivery}}
\end{figure*}
\subparagraph{Our model.} We are given a connected, undirected graph $G=(V,E)$ on $n:=|V|$ nodes with the length of each edge, together with
$m$ packages, each of which is given by a specified source-target node pair $(s_j, t_j)$. 
Furthermore there are $k>1$ mobile agents initially located at distinct nodes of the graph.
Each agent $i$ starts at position $p_i$ and can carry at most one package at a time. 
An agent can pick up a package at some node $u$ and drop the package again at some other node $v$.
In that sense it is also possible to hand over a package $j$ between agents if agent $a$ drops it at a node where another agent $b$ later picks up the same package, 
in which case we say that the agents \emph{collaborate} on package $j$.
A scenario in which each package $j$ is delivered from its source $s_j$ to its target $t_j$ is called a \emph{feasible solution} of the delivery problem.
Such a feasible solution $x$ specifies the \emph{travel distance} $d_i(x)$ that each agent $i$  has to travel in this solution.  

Each agent consumes energy proportional to the distance it travels in the graph, and we are interested in optimizing the total energy consumption for the team of agents.
Specifically, we consider heterogeneous agents with different rates of energy consumption (weights~$w_i$), and therefore the energy spent by agent $i$ in solution $x$ is
$w_i \cdot d_i(x)$, while the total energy of a solution is (see e.g., Figures~\ref{fig:delivery},\ref{fig:collabcost}): 
$\cost(x,w) := \sum_{i=1}^k w_i \cdot d_i(x)$, where $w = (w_1, \ldots, w_k)$.
We denote by $\OPT$ a solution with minimum cost among all solutions, $\OPT \in \argmin_x \cost(x,w)$.
It is natural to consider the scenario in which every agent is \emph{selfish} and therefore cares only about its own cost (energy) and not about the optimum social cost (total energy).

\subparagraph{Mechanism design.}
\label{sec:mechanism-design}
We consider the scenario in which agents can cheat or speculate about their costs: each $w_i$ is \emph{private} piece of information known to agent $i$ who
can report a possibly different $w_i'$. For instance, an agent may find it convenient to report a very high cost, in order to induce the underlying algorithm to assign a shorter travel distance to it (which may not be globally optimal). To deal with such situations, we introduce suitable compensations for the agents to incentivize them to truthfully report their costs. This combination of an algorithm and a payment rule is called a \emph{mechanism}, 
and we are interested in the following:
\begin{itemize}
	\item	\emph{Optimality.} The mechanism runs some (nearly) optimal algorithm such that, if agents do not misreport, 
		 the computed solution  has an optimal (or nearly optimal) social cost $\cost(x,w)$  with respect to the true weights $w$.
	\item	\emph{Truthfulness.} For every agent, truth-telling is a \emph{dominant strategy}, i.e., in no circumstance it is
		beneficial for an agent to misreport its cost. 	
		Thus, independently of the other agents, its utility (payment minus incurred cost)
		is maximized when reporting the true cost. 
	\item	\emph{Frugality.} The total payment to the agents should be nearly optimal, that is, comparable to the total energy cost. 
		Since agents should be payed at least their own cost when
		truthfully reporting (\emph{voluntary participation}), the mechanism \emph{must} pay 
		at least $\cost(x,w)$.
\end{itemize}
Intuitively speaking, we are paying the agents to make sure that they reveal their true costs, and, in this way, we can find a (nearly) optimal solution (w.r.t. the sum of weighted travel distances). 
At the same time, we want to minimize our total payments, i.e., we would like not to spend much more than the actual cost (weighted travel distance) of the solution.

\subparagraph{Our results.}
After some preliminaries on mechanism design and on energy-efficient delivery in Section~\ref{sec:preliminaries}, we first investigate in Section~\ref{sec:mechanism-full} 
the constant-factor $\left(4\cdot \max \smash{\tfrac{w_i}{w_j}}\right)$-approximation algorithm presented in~\cite{STACS17}. 
We reason why this algorithm cannot be turned into a mechanism that is both truthful \emph{and} guarantees voluntary participation.
However, using the algorithm as a black box for a new algorithm $A^*$ and applying Clarke's pivot rule for the payments to the agents, 
we give a truthful mechanism based on $A^*$ which guarantees voluntary participation and incurs a total energy of at most 
$\left(4\cdot \max\smash{\tfrac{w_i}{w_j}}\right)$ times the energy cost of an optimal delivery.

In Section~\ref{sec:mechanisms-constant}, we consider instances where either the number of packages $m$ is constant or the number of agents $k$ is constant. 
For both cases we provide constant-factor approximation algorithms with approximation factors that are independent of the agents' weights.
Both of these approximation algorithms satisfy sufficient conditions to be turned into truthful mechanisms with voluntary participation. 
The running time of the former algorithm can be improved to yield a $\mathit{FPT}$-approximation mechanism, parametrized by the number of packages; the latter runs in exponential time with $k$ as the base of the exponential term.

Finally, in Section~\ref{sec:frugality}, we discuss the frugality of mechanisms for the  case of a single package. In particular, we consider two truthful mechanisms, namely, the optimal one (which possibly uses several agents), and the one which always uses a single agent only.
Although the latter results in a higher energy cost, it might need a smaller sum of payments.
However, under some assumptions on the input, the payments of both mechanisms are only a small multiplicative factor larger than the minimum necessary (the cost of the optimum).

\subparagraph{Related work.} 
\emph{Energy-efficient} delivery has not been studied until recently.  
Most previous results are based on a model 
where the agents have uniform rates of energy-consumption but limited battery~\cite{DDalgosensors13}.
This restricts the possible movements of the agents -- one gets the decision problem of whether the given packages can be delivered without exceeding the available battery levels. 
This turns out to be \NP-hard even for a single package~\cite{DDicalp14,sirocco16} and even if energy can be exchanged between the agents~\cite{EnergyExchange15}.
The model of unbounded battery but heterogeneous weights has been introduced recently~\cite{STACS17} (for the full version see~\cite{BoC-TR16}).
Besides the mentioned approximation algorithm,
it was shown that a restricted solution in which each package is delivered by a single agent approximates an optimum delivery 
(in which agents can handover a package to another agent) by a factor of at most $2$.
This result has been named the \emph{Benefit of Collaboration}. 
Furthermore, the problem is \NP-hard to approximate to within a small constant, even for a single agent.

Approximating the \emph{maximum} travel distance of $k$ agents has been studied for other tasks such as visiting a set of given arcs~\cite{Frederickson76}, or visiting all nodes of a tree~\cite{FraGKP04}.
Furthermore Demaine et al.~\cite{Demaine2009} studied fixed-parameter tractability for
minimizing both the \emph{sum of} as well as the \emph{maximum} travel distance of agents for several tasks such as pattern formation, 
parametrized by the number of agents~$k$. These models consider only unweighted agents.

The problem of performing some collaborative task using \emph{selfish agents} is well studied in algorithmic game theory and, in particular, 
in algorithmic \emph{mechanism design} \cite{Nisan2001} where the system pays the agents in order to make sure that they report their costs truthfully. 
The existence of computationally feasible truthful mechanisms is one of the central questions, as truthfulness is often obtained by running an exact algorithm \cite{NisRon07}. 
In addition, even for simple problems, like shortest path, the mechanism may have to pay a lot to the agents \cite{Elk04,ArcTar07}.  
Network flow problems have been studied for the case when selfish agents own edges of the network and the mechanism pays them in order to deliver packages~\cite{Agarwal2008,Kalai82}. 
A setting where the transported packages are selfish entities which choose from fixed-route transportation providers was recently studied in \cite{fotakis2017selfish}.
In a sense, this is the reverse setting of our problem.

\section{Preliminaries}
\label{sec:preliminaries}
\subparagraph{Mechanisms.}

A mechanism is a pair $(A,P)$ where $A$ is an algorithm and $P$ a payment scheme which, for a given vector $w'=(w_1',\ldots,w_k')$ of costs reported by the agents, 
computes a solution $A(w')$ and a payment $P_i(w')$ for each agent $i$.
\begin{definition}[Truthful mechanism]
	A mechanism $(A,P)$ is \emph{truthful} if truth-telling is a dominant strategy (utility maximizing) for all agents. That is, for any vector $w'=(w_1',\ldots,w_k')$ of costs reported by the agents, for any $i$, and for any true cost $w_i$ of agent $i$,
	$P_i(w') - w_i \cdot d_i(A(w')) \leq P_i(w_i,w'_{-i}) - w_i \cdot d_i(A(w_i,w'_{-i}))$,
	where $d_i(x)$ is the travel distance of agent $i$ in solution $x$, and where $w'_{-i}:= (w_1',\ldots,w'_{i-1}, w'_{i+1},\ldots, w_k')$ and $(w_i,w'_{-i}):= (w_1',\ldots,w'_{i-1}, w_i, w'_{i+1},\ldots, w_k')$.
\end{definition}

Truthfulness can be achieved through a construction known as VCG mechanisms \cite{Vic61,Cla71,Gro73}, which requires that the underlying algorithm satisfies certain `optimality' conditions:

\begin{definition}[VCG-based mechanism]\label{def:VCG-payments}
	A~VCG-based mechanism is a pair $(A,P)$ of the following form: 
	For any vector  $w'=(w_1',\ldots,w_k')$ of costs reported by the agents, and for each agent $i$, there is a function $Q_i()$ independent of $w_i'$ such that $i$ is payed an amount of
	\begin{equation}
	\label{eq:VCG-payment}
	P_i(w') = Q_{i}(w'_{-i}) - \left(\sum\nolimits_{j\neq i} w'_j \cdot d_j(A(w'))\right).
	\end{equation}
\end{definition}

Intuitively speaking, these mechanisms turn out to be truthful, whenever the underlying algorithm minimizes the social cost with respect to a fixed subset of solutions 
(in particular, an optimal algorithm always yields a truthful mechanism):

\begin{theorem}[Proposition~3.1 in \cite{NisRon07}]\label{th:VCG-based}
	A VCG-based mechanism $(A,P)$ is truthful if  algorithm $A$ minimizes the social cost over a fixed subset $R_A$ of solutions. That is, there exists $R_A$ such that, for every $w'$,
	$A(w') \in  \argmin_{x\in R_A} \{ \cost(x,w') \}$.
\end{theorem}

The above result is a simple rewriting of the one in \cite{NisRon07} which is originally stated for a more general setting, in which agent $i$ valuates a solution $x$ by an amount $v_i(x)$, and $v_i()$ is the private information. Our setting is the special case in which these valuations are all of the form $v_i(x)= -w_i \cdot d_i(x)$ and $w_i$ is the private information (this setting is also called one-parameter \cite{ArcTar01}).

\begin{definition}[Voluntary participation]
	A mechanism $(A,P)$ satisfies the \emph{voluntary participation} condition if truth-telling agents have always a nonnegative utility. 
	That is, for every $w'=(w_1',\ldots,w_k')$, and for every agent $i$,
	$P_i(w') - w_i' \cdot d_i(A(w')) \geq 0$.
\end{definition}

We assume there are at least \emph{two agents} (otherwise the problem is trivial and there is no point in doing mechanism design; 
also one can easily show that voluntary participation \emph{and} truthfulness cannot be achieved in this case).
Voluntary participation can be obtained by the standard Clarke pivot rule, 
setting in the payments \eqref{eq:VCG-payment} the functions $Q_i()$ as follows:
\begin{equation}\label{eq:VCG:VP}
Q_i(w'_{-i}) := \cost(A(\bot,w'_{-i}),w'_{-i})
\end{equation}
where $(\bot,w'_{-i})$ is the instance in which agent $i$ is not present.
Note that, if algorithm $A$ runs in polynomial time, then the payments can also be computed in polynomial time: 
we only need to recompute $k$ solutions using $A$ and their costs. 
The next is a well-known result:

\begin{fact}\label{fac:VP}
	The VCG-based mechanism $(A,P)$ with the payments in \eqref{eq:VCG:VP} satisfies voluntary participation if the algorithm satisfies the following condition: For any vector $w'$ and for any agent $i$,
	$\cost(A(\bot,w'_{-i}),w')\geq \cost(A(w'),w')$.
\end{fact}
\begin{proof}
	Observe that the utility of agent $i$ is
	\begin{align*}
		Q_{i}(w'_{-i}) - \left(\smash{\sum\limits_{j\neq i}} w'_j \cdot d_j(A(w')) \mathclap{\phantom{2^{2^2}}} \right) - w_i \cdot d_i(A(w')) = Q_{i}(w'_{-i}) - \cost(A(w'),(w_i,w_{-i}')) & \\
	= \cost(A(\bot,w'_{-i}),w'_{-i}) - \cost(A(w'),(w_i,w_{-i}')). &
	\end{align*}
	When $i$ is truth-telling we have $w_i'=w_i$, and $w'=(w_i,w_{-i}')$. 
	Also, $\cost(A(\bot,w_{-i}'),w_{-i}')=\cost(A(\bot,w_{-i}'),w')$ because $i$ is not present in solution $A(\bot,w_{-i}')$.
	Hence the utility of $i$ is $\cost(A(\bot,w'_{-i}),w') - \cost(A(w'),w')$, which is non-negative by assumption.
\end{proof}

\begin{definition}[Approximation mechanism]
	A mechanism $(A,P)$ is a $c-$approximation mechanism, if for every input vector of bids $w'$ its algorithm $A$ computes 
	a solution $A(w')$ which is a $c-$approximation of a best solution $\OPT(w')$, i.e., $\cost(A(w'),w') \leq c\cdot \cost(\OPT(w'),w')$.
\end{definition}

\subparagraph{Collaboration of agents.}

\begin{figure*}[t!]
	\includegraphics[width=\linewidth]{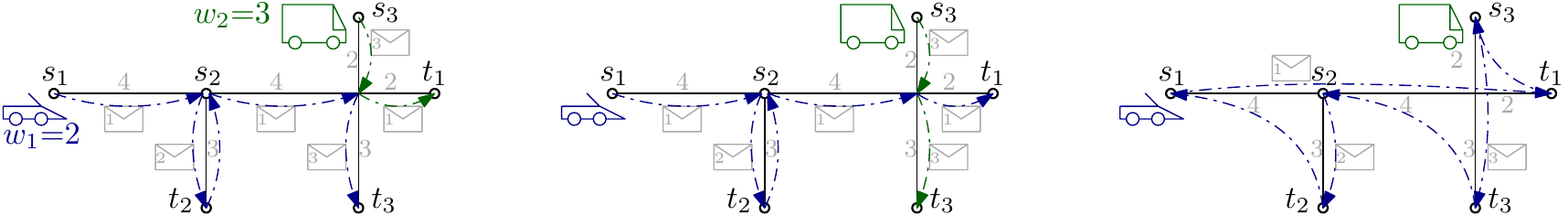}
	\caption{(left) Optimal solution of cost $2\cdot 17 + 3\cdot 4 = 46$ with collaboration on packages $1$ and $3$.\newline 
	(middle) Optimal solution among all non-collaborative solutions with energy cost $2\cdot 16+3\cdot 5 = 47$.\newline
	(right) Optimal non-collaborative solution \emph{with} direct delivery and return; total energy $2\cdot 2\cdot 18 = 72$.\label{fig:collabcost}}	
\end{figure*}

To describe a solution for an instance of the delivery problem, we can (among other characteristics) elaborate on the following properties of the solution: 
How do agents work together on each package (\emph{collaboration}), how are agents assigned to packages (\emph{coordination}) and which route does each agent take (\emph{planning}).
Most of the mechanisms in this paper are based on the characterizations of the \emph{benefit of collaboration}:
 
\newcommand{\res}{R^{\star}_{\noC}}
\begin{definition}[Benefit of collaboration]\label{def:boc}
	Define $R_{\noC}$ as the set of solutions $x$ in which there is no collaboration of the agents,
	meaning that each package is delivered by a single agent only.
	Define $\res \subset R_{\noC}$ as the subset of solutions $x\in R_{\noC}$ in which
	(i) every package is transported directly from source to target without intermediate dropoffs, and
	(ii) every agent returns to its respective starting location.
	The ratios $\BoC := \smash{\min\limits_{x\in R_{\noC}} \tfrac{\cost(x,w)}{\cost(\OPT,w)}}$
	and $\BoC^* := \smash{\min\limits_{x\in \res} \tfrac{\cost(x,w)}{\cost(\OPT,w)}}$
	are called \emph{benefit of collaboration}.
\end{definition}
\smallskip

\begin{theorem}[Theorems~5, 6 \& 7 in \cite{STACS17}]
	\label{thm:boc}
	For the the benefit of collaboration we have $\BoC \leq 1/\ln 2$ for a single package ($m=1$, for a tight lower bound example see Fig.~\ref{fig:delivery}), 
	and $\BoC \leq \BoC^* \leq2$ in general ($m\geq1$, see for example Fig.~\ref{fig:collabcost}).
\end{theorem}

\section{Truthful Approximation Mechanisms}

In this section, we present two polynomial-time truthful mechanisms.
Note that energy-efficient delivery is \NP-hard to approximate to within any constant approximation ratio less than 367/366 \cite[Theorem 9]{STACS17}.
Hence, even in the best case, our goal can only be to guarantee a \emph{constant-factor approximation} of the optimum in a truthful way.

In the first part we present a polynomial-time truthful approximation mechanism with a constant-factor approximation guarantee 
\emph{depending on the weights} of the agents. In the second part we turn to \emph{absolute constant}-factor approximation, 
for which we require the number of packages $m$ to be constant in order to get a polynomial-time mechanism. 
Similar techniques yield a mechanism running in exponential time for a constant number of agents $k$. 

\subparagraph{Main algorithmic issue.}
The two truthful polynomial-time approximation mechanisms we obtain are based on the following scheme (and the third approximation mechanism follows the same scheme but takes exponential time):
\begin{enumerate}
	\item	Precompute a feasible subset $R$ of solutions \emph{independently of the input weights} $w'$. This set may depend on the name/index and on the position of the agents. 
	\item	Among all precomputed solutions in $R$, return the best solution with respect to the input weights $w'$, that is, a solution
		$x^*\in  \argmin_{x\in R} \{ \cost(x,w') \}$.
\end{enumerate}
Truthfulness then follows directly by Theorem~\ref{th:VCG-based}. 
Note that, since we want polynomial running time, the first step selects a \emph{polynomial} number of solutions (thus the second step is also polynomial). 
The main crux here is to make sure that, for all possible input weights $w'$, the set $R$ contains at least one \emph{good approximation} 
(a solution $x\in R$ whose cost for $w'$ is at most a constant factor above the optimum for $w'$).

\subsection{A polynomial-time approximation mechanism}
\label{sec:mechanism-full}

We start with the general setting of arbitrarily many packages $m$ and arbitrarily many agents $k$. 
Our construction of a truthful approximation mechanism $(A^*,P)$ for this setting relies on Theorem~\ref{th:VCG-based}. 
To this end, we define a fixed subset of solutions $R_{A^*}$ and a polynomial-time algorithm $A^*$, 
such that for every vector of reported weights $w'$, the algorithm $A^*$ computes a solution $S$ of optimum cost among all solutions in $R_{A^*}$, 
$S \in \argmin_{x\in R_{A^*}} \cost(x,w')$. 

In a next step, we show that whenever the agents report truthfully ($w' = w$), the computed solution $x\in R_{A^*}$ has an energy cost that approximates 
the overall optimum energy cost by a constant factor of at most $4\cdot \tfrac{w_{\max}}{w_{\min}}$, where $w_{\max} := \max_i{w_i}$ and $w_{\min} := \min_i{w_i}$.

\newcommand{\Apos}{A_{pos}}  
\subparagraph{A first approach.} We first analyze an existing approximation algorithm $\Apos$ presented in~\cite{STACS17}, 
which computes a solution depending only on the position of the agents, but not on their reported costs.
Roughly speaking, $\Apos$ connects the agent positions and the package sources and targets by a minimum spanning forest, subject to the following:
In each tree, there is (i) an edge between the corresponding source/target nodes $s_j,t_j$ and (ii) exactly one agent position $p_i$. All edge lengths correspond 
to the distances in the original instance. Agent $i$ then traverses its tree in a DFS-like fashion, crossing each edge twice and thus delivering all packages in the tree. 
For an example of such a solution, see Figure~\ref{fig:approxfull} (left).

\begin{theorem}[Theorem~13 in \cite{STACS17}]\label{th:BoC-apx}
	For any number of packages $m$ and agents $k$, 
	there exists a polynomial-time algorithm $\Apos$ which computes 
	a $\left(\smash{4\cdot \frac{w_{\max}}{w_{\min}}}\right)$-approximation. 
	The computed solution does not depend on the input weights $w$, only on the position of the agents.
\end{theorem}

\noindent However, there is no mechanism $(\Apos, P)$ which is both truthful \emph{and} guarantees voluntary participation:
Assume for the sake of contradiction there is such a mechanism. 
To guarantee voluntary participation we need the payments to satisfy $P_i(w') - w_i\cdot d_i(\Apos(w')) \geq 0$ for each agent $i$.
Consider any agent $i$ which is used in the solution, i.e., $\smash{d_i(\Apos(w_i, w_{-i}')) > 0}$.
Note that $\Apos$ selected $i$ independent of its weight $w_i$, hence $d_i$ remains the same for all reported weights $w_i'$, i.e., we have  
$d_i := d_i(\Apos(w_i,w_{-i}')) = d_i(\Apos(w_i',w_{-i}'))$.
Let $P_i(w_i,w_{-i}')$ denote the payment to agent $i$ when she reports her true weight $w_i$.
Now consider a situation where agent $i$ reports a different value $w_i' > P_i(w_i,w_{-i}')/d_i$ instead:
For voluntary participation of agent $i$, the payment $P_i(w')$ needs to satisfy 
$P_i(w') - w_i'\cdot d_i \geq 0 \Leftrightarrow P_i(w') \geq w_i' \cdot d_i$, but $w_i'\cdot d_i >  P_i(w_i,w_{-i}')$, contradicting the truthfulness of the mechanism 
(in other words: reporting an arbitrary high weight results in an arbitrary high payment).

\begin{algorithm}[t!]
\caption{$A^*$}
\begin{algorithmic}[1]
\Input Connected graph $G$, $k$ agents, $m$ packages, black box algorithm $\Apos$ (Thm.~\ref{th:BoC-apx}).
\Output A solution $S$ with cost at most the cost of $\Apos$. 	

\State	Compute the following $k+1$ solutions using $\Apos$ as a black box subroutine:
	\label{algo:step1}
	\[
		x_0:= \Apos(w'),\text{ and }\forall i=1,\ldots,k\colon x_{-i}:=\Apos(\bot,w_{-i}').
	\]
	All solutions in the set $R_{A^*}:=\{x_0, x_{-1},\ldots,x_{-k}\}$ are feasible by connectivity of $G$.
\State	Define algorithm $A^*$ as taking the best among all these solutions with respect to the input weights $w'$:
	\label{algo:step2}
	\[	
		A^*(w'):= \argmin_{x \in \smash{R_{A^*}}}\{\cost(x,w')\}\enspace .
	\]
\end{algorithmic}
\end{algorithm}

\subparagraph{Refining the approach.} In order to obtain a truthful mechanism \emph{with} voluntary participation we consider the algorithm $A^*$ 
obtained from $\Apos$ and a repeated application of algorithm $\Apos$ on all subsets of $(k-1)$ agents 
(for feasibility recall that $G$ is connected and that $k>1$).

\begin{theorem}\label{th:mechanism-general}
	There exists a polynomial-time truthful VCG mechanism $(A^*,P)$ satisfying voluntary participation 
	with approximation ratio of at most the approximation ratio of $\Apos$.
\end{theorem}

\begin{proof}
	We use VCG payments \eqref{eq:VCG-payment} with \eqref{eq:VCG:VP} as $Q_i(w'_{-i}) := \cost(\Apos(\bot,w'_{-i}),w'_{-i}) = \cost(x_{-i},w'_{-i})$.
	Since in Step~\ref{algo:step1} every solution is computed independently of the input weights, 
	algorithm $A^*$ satisfies the conditions of Theorem~\ref{th:VCG-based}, which implies truthfulness. 
	We show voluntary participation along the lines of Fact~\ref{fac:VP}: First note that by Step~\ref{algo:step2} in $A^*$:
	\begin{align}
		\cost(A^*(w'),w')	& \leq 
		\cost(\Apos(\bot,w'_{-i}),w')\enspace, \label{new}
	\end{align}
	since the solution $x_{-i} = \Apos(\bot,w'_{-i})$ is a feasible solution that $A^*$ considers on input $w'$.
	By the same argument, the approximation ratio of $A^*$ is at most the approximation ratio of $\Apos(w)$.
	For voluntary participation, note that for $w_i' = w_i$, $i$'s utility is
	\begin{align*}
		Q_{i}(w'_{-i}) - \left(\vphantom{2^{2^2}}\smash{\sum_{j\neq i}} w'_j \cdot d_j(A^*(w'))\right) - w_i \cdot d_i(A^*(w')) = Q_{i}(w'_{-i}) - \cost(A^*(w'),(w_i,w_{-i}')) & \\
	= \cost(\Apos(\bot,w'_{-i}),w'_{-i}) - \cost(A^*(w'),(w_i',w_{-i}')) \stackrel{\eqref{new}}{\geq} 0.\quad \qedhere 
	\end{align*}
\end{proof}

For an illustration of the computed set $R_{A^*}$, see Figure~\ref{fig:approxfull}: 
The mechanism $(A^*,P)$ will pick solution $x_{-2}$ and award payments $P_1 = 46-(10-4) = 40,\ P_2 = 0, P_3 = 160-(10-6) = 156$.
\begin{figure}[t!]
	\includegraphics[width=\linewidth]{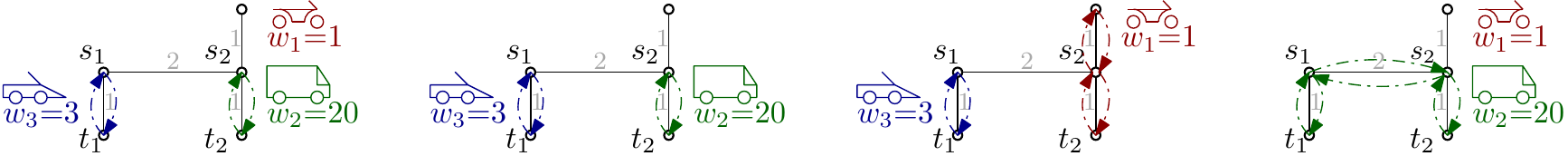}
	\caption{(from left to right)
		Original solution $x_0 = \Apos$ with energy cost $\cost(x_0,w) = 46$,
		solutions $x_{-1}, x_{-2}, x_{-3}$ with energy $\cost(x_{-1},w) = 46$, $\cost(x_{-2},w) = 10$ and $\cost(x_{-3},w) = 160$.\label{fig:approxfull}}		
\end{figure}

\subsection{Absolute constant-factor approximation mechanisms}

\label{sec:mechanisms-constant}

We now turn to developing truthful mechanisms where the approximation guarantee will be \emph{an absolute constant},
therefore independent of the weights (cf.~Theorems~\ref{th:BoC-apx} and \ref{th:mechanism-general}).
To this end, we provide a polynomial-time truthful approximation mechanism if the number of packages $m$ is constant.
We also show that by slightly deviating from the 2-step scheme given in the beginning of this section, 
the running time of the algorithm can be improved to $f(m)\cdot(k n)^{\bigO(1)}$, where $f$ depends only on $m$.
Hence we get a truthful fixed parameter tractable (FPT) approximation mechanism, parametrized by the number of packages.
To the best of our knowledge, the underlying algorithm is also the first absolute constant-factor approximation for the delivery problem.
For the case of only constantly many agents $k$, we provide an exponential-time algorithm, where $k$ is the base of the exponential term.

\subparagraph{No collaboration.} In the following, we consider only solutions $x\in \res$ where agents do not collaborate and follow the two properties given in Theorem~\ref{thm:boc}.
We are therefore left with the tasks of \emph{coordinating} which agent gets assigned to which packages and \emph{planning} in which order she delivers the assigned packages. 
In other words, in every such solution $x \in \res$, each agent $i$ is assigned a (possibly empty) block of packages $M_i(x) = \{ i_1, i_2, \ldots, i_{|M_i(x)|} \}$.
These package blocks are disjoint and form a partition of all $m$ packages $\{1,\ldots, m\}$. 
Furthermore, agent $i$ delivers each of its packages directly after picking it up at its source. These packages are processed in some order $\pi_x(i_1),\pi_x(i_2),\ldots, \pi_x(i_{|M_i(x)|})$, where $\pi_x$ is a permutation of $\{1,\ldots, m\}$. Finally, $i$ returns to its starting location $p_i$.
Therefore, denoting the distance between $u$ and $v$ in $G$ by $\dist(u,v)$, the travel distance $d_i(x)$ of agent $i$ can be written as
\begin{eqnarray}
	d_i(x) =& \dist\left(p_i, s_{\pi_x(i_1)}\right) + \smash{\sum\limits_{j=1}^{|M_i|}} \dist\left(s_{\pi_x(i_j)},t_{\pi_x(i_j)}\right) + \nonumber\\
	& \sum\limits_{j=1}^{\mathclap{|M_i|-1}} \dist\left(t_{\pi_x(i_j)},s_{\pi_x(i_{j+1})}\right) + 
	\dist\left(t_{\pi_x(i_{\smash{|M_i|}})},p_i\right). \label{eq:minperm}
\end{eqnarray}

\begin{myremark}\label{rem:dist}
	In the $\dist(u,v)$ terms in \eqref{eq:minperm} we are not interested in the actual route agent $i$ takes, as long as it uses a shortest path between $u$ and $v$.
\end{myremark}

\subparagraph{Sets and lists.} To clarify the use of package blocks and package orders, we use the standard notion of sets and lists regarding partitions (see e.g.,~\cite{Callan08}):
If we look at a partition of $\left\{ 1,\ldots, m \right\}$ into non-empty disjoint blocks, we can take into account the order of the elements within blocks, 
the order of the blocks, or both. We get four cases: sets of sets, sets of lists, lists of sets and lists of lists (\cite{OESISsetsofsets,OESISsetsoflists,OESISlistsofsets,OESISlistsoflists}).
However, in the delivery setting we can also have agents which are not used at all and therefore not assigned to any packages -- a complete description (\emph{coordination + planning})
of a solution $x\in \res$ is therefore given by a \emph{list of exactly $k$ (possibly empty) lists},
$M=(M_1,\ldots,M_k)$,
where list $M_i$ represents the sequence of the packages that agent $i$ has to deliver in the order specified by $M_i$. Hence, we immediately get a bijection from all lists of exactly $k$ possibly empty lists to $\res$ (modulo the equivalence between shortest paths -- see Remark~\ref{rem:dist}). 
\begin{fact}\label{fac:list-solution}
	For every such list of lists $M$, there is a  solution $x_M \in \res$ 
	in which each agent $i$ delivers the packages in $M_i$ in the order specified by this list and 
	in which the cost is minimized. Given $M$, such a solution can be computed in time $\mathit{poly}(n,m,k).$ 
\end{fact}

\begin{algorithm}[t!]
\caption{$A^m$ (for a constant number $m$ of packages)} 
\begin{algorithmic}[100]
\Input  Connected graph $G$, $k$ agents, $m$ packages.
\Output An optimal solution $S\in \res$.
\State	\hspace{-3ex}1:\ Brute-force enumeration over all lists of exactly $k$ possibly empty lists of the packages. 
	\Foreach{list of $k$ lists} 
		\State	Add the corresponding solution (Fact~\ref{fac:list-solution}) to the set of solutions~$R_{A^m}$.
	\EndForeach
\State	\hspace{-3ex}2:\ Define algorithm $A^m$ as taking the best among all solutions in $R_{A^m}$ with respect to the input weights $w'$:
	\[	
		A^m(w'):= \argmin_{x \in \smash{R_{A^m}}}\{\cost(x,w')\}\ .
	\]
\end{algorithmic}
\end{algorithm}

\subsubsection*{Constant number of packages $m$}

We now look at the case of a constant number $m$ of packages. 
By Theorem~\ref{thm:boc}, the solution $S:=\argmin_{x\in \res} {\cost(x,w)}$ is a $2-$approximation of the optimum.  
First, we present an algorithm $A^m$ which basically enumerates over all \emph{lists of exactly $k$ (possibly empty) lists}
and adds the corresponding solution to a set $R_{A^m}$ (where we get $R_{A^m} = \res$), as described in the first step of our scheme.
Then, given an input vector of weights $w'$, $A^m$ chooses the best solution $A^m(w') \in \arg\min_{x\in \res} \left\{ \cost(x,w' \right\}$.

\begin{theorem}
	\label{thm:algm}
	Algorithm $A^m$ 
	finds a best solution $S \in$ $\smash{\argmin\limits_{x\in \res} {\cost(x,w)}}$ and can be implemented to run 
	in time $\bigO\left( m!(k+m)^m \cdot \mathit{poly}(n,m,k)\right)$.
\end{theorem}

\begin{proof}
	Step 1 of $A^m$ produces $\bigO(m! \cdot {m+k-1 \choose k-1})$ many solutions and can be implemented in time $\bigO\left( m!(k+m)^m \cdot \mathit{poly}(m,k)\right)$ as follows: 
	Since we look at a list of lists we have a total order on the packages. Hence we first enumerate in an outer loop over all $m!$ permutations, which can be done in time $\bigO(m!)$.
	Each such permutation also needs to be subdivided into exactly $k$ possibly empty lists. There are ${m+k-1 \choose k-1} \leq (k+m)^m$ many ways to do this (by placing $k-1$ delimiters at
	$m+k-1$ potential positions in time $\bigO\left(\mathit{poly}(m,k)\right)$).
	Step 2 of $A^m$ consists of computing the cost of each solution $x\in \res$ 
	in time $\bigO \left( \mathit{poly}(n,m,k)\right)$. 
\end{proof}

\begin{theorem}\label{thm:mechanism-constantm}
	For a constant number of packages $m$, there exists a polynomial-time truthful VCG mechanism $(A^m,P)$, satisfying voluntary participation,
	with approximation ratio $\leq 2$.
\end{theorem}
\begin{proof} 
	Theorem~\ref{thm:algm} says that the algorithm satisfies the conditions in  Theorem~\ref{th:VCG-based}, and thus truthfulness holds.  Theorem~\ref{thm:algm} also implies the running time. The approximation is due to Theorem~\ref{thm:boc} and the definition of the benefit of collaboration $\BoC^*$. 
	We next argue that the algorithm satisfies the condition in Fact~\ref{fac:VP}, $\cost(A^m(w'),w') \le \cost(A^m(\bot,w'_{-i}),w')$.
	Indeed, every solution $A^m(\bot,w'_{-i})$ is also considered by the algorithm when all agents are present (input $w'$), since this solution corresponds to some list of $k$ lists $M$ 
	in which agent $i$ is not given any package (i.e.~$M_i=\emptyset$). 
	Thus the solution $A^m(\bot,w'_{-i})$ is contained in $\res$.
\end{proof}

We now show that the running time of the mechanism $(A^m,P)$ can be improved. The main idea is to enumerate in $A^m$ over all \emph{sets of lists} instead of \emph{list of lists} -- 
i.e. during the enumeration we do not fix yet which agent gets which list of packages.
Rather for each set of lists, we aim to directly compute an optimal assignment between the agents and the lists, 
thus deciding for every fixed set of lists in a \emph{parallel way} the assigned list for every agent 
(instead of enumerating over all possible assignments as well). 

\begin{algorithm}[t!]
\caption{$A^m$ (improved version)} 
\begin{algorithmic}[100]
\Input  Connected graph $G$, $k$ agents, $m$ packages.
\Output An optimal solution $S\in \res$.
\State	Enumerate over all \emph{sets} of (non-empty) lists of the packages $1,\ldots, m$. 
	\Foreach{set of $\leq k$ lists}
		\Foreach{pair (agent $i$, list $M_j$)}
			\State 	\hspace{-3ex}1a:\ Assume agent $i$ delivers the packages in $M_j$ in their order.
			\State 	\hspace{-3ex}1b:\ Compute the cost $d_i(M_j)$ of doing so. 			
		\EndForeach
		\State	\hspace{-4ex}2a: Build a complete bipartite graph Agents--Lists 
		\State	with edge costs $w_i\cdot d_i(M_j)$ for each edge $\left\{i,M_j\right\}$.
		\State	\hspace{-4ex}2b: Find the best assignment Agents $\rightarrow$ Lists 
		\State	(by computing a maximum weighted bipartite matching).
		\Foreach{subset of $k-1$ agents}
			\State	\hspace{-3ex}3: Repeat 2a, 2b for the subset of $k-1$ agents.
		\EndForeach
		\State	\hspace{-3ex}4: Keep track of the best solution(s) found so far.
	\EndForeach
\end{algorithmic}
\end{algorithm}

\begin{theorem}\label{thm:mechanism-constantm-improved}
	The running time of the truthful VCG mechanism $(A^m,P)$ can be improved to show fixed parameter tractablility, parametrized by the number of packages, namely a running time of
	$\bigO(f(m) \cdot \mathit{poly}(n,k))$, where $f(m) \in \smash{\bigO\left(e^{2\sqrt{m}-m} m^m\cdot \mathit{poly}(m) \right)}$.
\end{theorem}

\begin{proof}
	Consider first the running time of the improved algorithm $A^m$, which iterates over all sets of non-empty lists of the packages $1,\ldots, m$. 
	The number of sets of lists is known to be $\bigO(e^{2\sqrt{m}-m} m^m / \mathit{poly}(m))$~\cite{OESISsetsoflists}. 
	To this end we enumerate over all sets of lists while spending only an additional $\bigO(\mathit{poly}(m))$-factor (over the number of sets of lists) on the running time.
	This can be done by enumerating over all \emph{sets of sets} (by considering packages one-by-one and deciding whether
	to put them in a previously created subset or whether to start a new subset), followed by enumerating over all permutations of the packages inside each subset.

	For each of the sets of at most $k$ lists, we compute the best assignment of the $k$ agents to the lists (say $l\leq k$ many) with a weighted bipartite matching as follows:
	On one side of the bipartite graph, we have $k$ vertices, one for each agent. On the other side, we have one vertex per list. 
	We take the complete bipartite graph between the two sides. This graph has at most $k\cdot l \leq k \cdot m$ edges. 
	We compute the cost of an (agent, list)-edge as the energy cost of delivering all packages in the bundle in that order with that agent.
	This requires $\bigO(\mathit{poly}(n,m))$ time per edge and $\bigO(\mathit{poly}(n,k,m))$ time in total. 
	A \emph{maximum matching of minimum cost} in this graph gives the best assignment of agents to bundles and can be found in polynomial time,
	e.g., by using the Hungarian method~\cite{kuhn1955hungarian} or the successive shortest path algorithm~\cite{edmonds1972theoretical}.

	It remains to show that truthfulness and voluntary participation also hold for the improved version of algorithm $A^m$. 
	Truthfulness follows immediately from the fact that $A^m$ still considers all solutions $x\in \res$, albeit always several of them in a parallel fashion.
	To be able to apply Clarke's pivot rule to define the payments via $Q_{i}(w_{-i}')$ we make sure to also consider all solutions on all $k$ different subsets of $k-1$ agents, see Step 3 of the improved algorithm $A^m$.
\end{proof}


\subsubsection*{Constant number of agents $k$}

Next, we look at settings with a constant number of agents $k$ but an arbitrary number of packages $m$. 
In this case, even for $k=1$ and independent of whether or not we restrict the packages to be transported from source to target without intermediate drop-offs, 
it is \NP-hard to approximate $\min_x\cost(x,w')$ to within any constant approximation ratio less than $367/366$~\cite[Theorem 9]{STACS17}.
There, the bottleneck lies in finding the optimal permutation $\pi_{\OPT}$ to minimize the travel distances, see Equation~\eqref{eq:minperm}.
For $k>1$, we first have to partition the packages into (possibly empty) subsets $M_1, \dots, M_k$.
Contrary to $A^*$ and $A^m$, we will need exponential time $\bigO(k^m)$ to enumerate all partitions.
Next, for every fixed partition of the packages into subsets $M_i$, we look for a package order
$\pi_x(i_1),\pi_x(i_2),\ldots, \pi_x(i_{|M_i|})$ given by a permutation $\pi_x$ of $\left\{ 1,2,\ldots,m \right\}$ such that we get an approximation guarantee on $d_i(x)$.

\subparagraph{Stacker-Crane problem.} The latter can be modeled as the \emph{Stacker-Crane problem}, which asks for the following: 
Given a weighted graph $G_{\mathit{SCP}}$ with a set of \emph{directed arcs} and a set of \emph{undirected edges}, find the minimum tour that uses each arc at least once.
Since we restrict ourselves to the two additional conditions (i) direct delivery of each package, (ii) return of each agent $i$ in the end,
we can model the transport of package $j$ along a shortest path by a directed edge from $s_j$ to $t_j$:
Hence we choose the graph $G_{\mathit{SCP}}$ to consist of node $p_i$ and all nodes $s_j, t_j,\ j = i_1,\ldots,i_{|M_i|}$, 
together with directed arcs $(s_j, t_j)$ of weight $\dist_G(s_j,t_j)$ (corresponding to the length of a shortest path between $s_j$ and $t_j$ in $G$)
and undirected edges $\smash{\left\{ t_{j1}, s_{j2} \right\}, \left\{ p_i, s_{j} \right\}, \left\{ p_i, t_j \right\}}$ of weights 
corresponding to the original distances $\dist_G(t_{j1}, s_{j2})$ between $t_{j1}, s_{j2}$ (respectively $\dist_G(p_i, s_j)$, $\dist_G(p_i, t_j)$). 
For the Stacker-Crane problem, a polynomial-time $1.8$-approximation due to Frederickson~et~al. is known~\cite{FraGKP04}.
It remains to iterate over all assignments of the packages to the $k$ agents as in the presented Algorithm~$A^k$.
\begin{theorem}
	\label{thm:algk}
	Algorithm $A^k$ runs in time $\bigO\left( \mathit{poly}(n,m,k)\cdot k^m \right)$
	and finds an approximate solution $S$ of $\cost(S,w') \leq1.8 \cdot \smash{\min_{x\in \res}}\cost(x,w')$.
\end{theorem}
\begin{algorithm}[t!]
\caption{$A^k$ (for a constant number $k$ of agents)} 
\begin{algorithmic}[100]
\Input  Connected graph $G$, $k$ agents, $m$ packages.
\Output A $1.8$-approximate solution $S\in \res$. 
\State	\hspace{-3ex}1:\ Brute-force enumeration over all lists of exactly $k$ possibly empty sets of the packages. 
	\Foreach{list of $k$ sets $(M_1, \ldots, M_k)$}
		\Foreach{agent $i$}
			\State	\hspace{-3ex}a: Model the delivery of the packages $M_i$ (by agent $i$) as a Stacker-Crane problem.
			\State	\hspace{-3ex}b: Compute a solution $x|_{M_i}$ such that $d_i(x)$ is a 1.8-approximation.
		\EndForeach
		\State	\hspace{-3ex}c: Add delivery problem solution $x$, combined from the $k$ Stacker-Crane solutions $x|_{M_i}$, 
		\State	to the set of solutions $R_{A^k}$.
	\EndForeach
\State	\hspace{-3ex}2:\ Define algorithm $A^k$ as taking the best among all solutions in $R_{A^k}$ with respect to the input weights $w'$:
	\[	
		A^k(w'):= \argmin_{x \in \smash{R_{A^k}}}\{\cost(x,w')\}\enspace .
	\]
\end{algorithmic}
\end{algorithm}
\begin{proof}
	Algorithm $A^k$ first enumerates over all lists of exactly $k$ possibly empty sets of the packages. 
	This can be implemented to run in time $\bigO(k^m)$ by choosing for each of the $m$ packages the subset of packages $M_i$ it belongs to.  
	We know that for each list of exactly $k$ possibly empty lists there is a corresponding solution $x\in \res$ and vice versa (Fact~\ref{fac:list-solution}). 
	In particular there exists a list $M=(M_1,\ldots, M_k)$ of lists $M_1, \ldots, M_k$ for each optimal solution $\OPT(\res) \in \arg \min_{x\in \res}\left\{ \cost(x,w') \right\}$.
	Algorithm $A^k$ at some point considers the list of exactly $k$ possibly empty \emph{sets} $M_1, \ldots, M_k$ (where there is no prescribed order of the elements in each of the $M_i,\ i=1,\ldots, k$).
	Applying the Stacker-Crane approximation algorithm, $A^k$ approximates each of the travelling distances $d_i(\OPT(\res))$ of the agents by a factor of at most $1.8$. 
	Hence $A^k$ also considers a solution $S' \in \res$ of cost
	\begin{align*}
		\cost(S',w') = \sum_{i=1}^k{w_i'\cdot d_i(S')}	&\leq \sum_{i=1}^k{w_i'\cdot 1.8\cdot d_i(\OPT(\res))} 
								= 1.8 \cdot \min\limits_{x\in \res}{\cost(x,w')}.
	\end{align*}
	Since this solution $S'$ is contained in the built set $R_{A^k}$, we know that $S:= A^k(w')$ has $\cost(S,w') \leq \cost(S',w') \leq 1.8 \cdot \min_{x\in \res}{\cost(x,w')}$.
\end{proof}

\begin{theorem}\label{thm:mechanism-constantk}
	For a constant number of agents $k$, there exists an \emph{exponential}-time truthful VCG mechanism $(A^k,P)$, satisfying voluntary participation,
	with approximation ratio $\leq 3.6$.
\end{theorem}

\begin{proof}
	The approximation follows from Theorem~\ref{thm:boc} and Theorem~\ref{thm:algk}. The latter theorem also implies the running time of the mechanism. Truthfulness and voluntary participation can be proved essentially in the same way as for Theorem~\ref{thm:mechanism-constantm}. Indeed, the set of solutions $R_{A^k}$ is computed independently of the input weights $w'$, and thus the last step defining $A^k$ satisfy the condition of Theorem~\ref{th:VCG-based} (implying truthfulness). As for voluntary participation, we observe that when $A^k$ is run on input $(\bot,w'_{-i})$, the computed solution corresponds to some list of $k-1$ sets (agent $i$ is not present), and the same solution is considered on input $w'$ as a list of $k$ sets, where one set is empty. This implies that the algorithm satisfies
	$\cost(A^k(w'),w') \le \cost(A^k(\bot,w'_{-i}),w')$, i.e., Fact~\ref{fac:VP} and thus voluntary participation.
\end{proof}

\subsection{Comparison of the algorithms}

\begin{figure}[b!]
	\includegraphics[width=\linewidth]{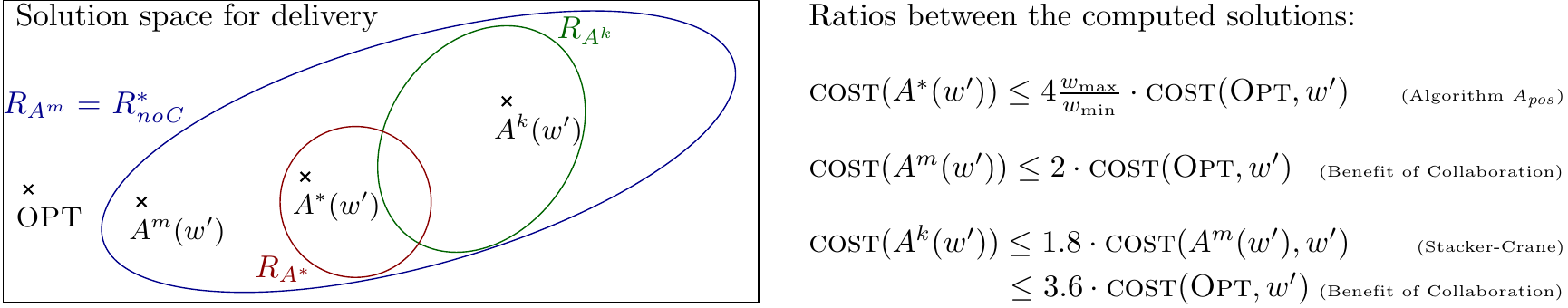}
	\caption{Venn diagram of the subsets of solutions $R_{A^*},\, R_{A^m},\, R_{A^k}$ considered by the given algorithms.
	Depending on the actual instance, the intersection of $R_{A^*}$ and $R_{A^k}$ might be empty.\label{fig:comparison}}
\end{figure}

We conclude our results on truthful approximation mechanisms with a comparison of the three given algorithms 
$A^*$ (general setting, polynomial time), $A^m$ (FPT, parametrized by the number of packages $m$) and $A^k$ (exponential time for a constant number of agents $k$),
given in Figure~\ref{fig:comparison}.
Each of the three algorithms chooses an optimal solution from a respective set $R_{A^*}, R_{A^m}, R_{A^k}$.  
To ensure voluntary participation, in each set we include solutions where individual agents are not present -- 
in this way, we can set the payments according to Clarke's pivot rule~\eqref{eq:VCG:VP}.
All three algorithms make use of Theorem~\ref{thm:boc} which bounds the Benefit of Collaboration $\BoC^*$ by 2 
(this can be seen directly in the definition of the algorithms $A^m, A^k$, 
for algorithm $A^*$ this follows from the black box algorithm $\Apos$~\cite[Theorem 13]{STACS17}).
In fact, the first version of the algorithm $A^m$ computes a set $R_{A^m}$ which consists of all the solutions in $\res$ -- 
the improved version of $A^m$ discards solutions of high cost early on (while still considering solutions where individual agents are not present), 
in order to limit the set $R_{A^m}$ to a small ($\mathit{FPT}$-) size, parametrized by the number of packages. 
Finally, we can compare the approximation guarantees. For algorithm $A^*$ the approximation ratio of $4\cdot \tfrac{w_{\max}}{w_{min}}$ follows from the (same) approximation ratio of $\Apos$. 
Both $A^m$ and $A^k$ have a factor of $\BoC^* \leq 2$, with an additional factor of $1.8$ for $A^k$ 
since we need to approximate the planning of each agent's tour (which we model with stacker-crane).

\section{Single Package and Frugality}
\label{sec:frugality}

For the case of a \emph{single package}, we define two truthful VCG mechanisms. The \emph{optimal} mechanism which minimizes the social cost 
using any number of agents, and the \emph{lonely} mechanism which computes the solution of minimal cost under the constraint of using only one single agent. In both mechanisms, we use the VCG payments \eqref{eq:VCG-payment} with Clarke pivot rule \eqref{eq:VCG:VP} in order to satisfy voluntary participation.

\begin{theorem}[Theorem 2 in~\cite{STACS17}]\label{th:BoC}
	For $m=1$, the optimal solution using a single agent, as well as the optimal solution using any number of agents, can be computed in polynomial time. 
\end{theorem}

\begin{fact}
	Both the exact and the lonely mechanisms are truthful since the algorithms are optimal with respect to a fixed subset of solutions, i.e., they satisfy the condition of Theorem~\ref{th:VCG-based}. Moreover, the mechanisms run in polynomial time by Theorem~\ref{th:BoC}.
\end{fact}
In the following, we bound the total payments of these mechanisms compared with the cost of the optimal solution (for the given input). In other words, assuming the reported weights are the true weights ($w'=w$), we would like the mechanism to not pay much more than the optimum for these weights $w=w'$. This property is usually termed \emph{frugality} \cite{Elk04,ArcTar07}.

We first observe that, if we care more about the total payment made to the agents
than about the optimality of the final solution, then in some instances it may pay off to run the lonely mechanism instead of the optimal mechanism. However, in other instances, the converse happens, meaning that neither mechanism is always better than the other.

In order to compare the optimal mechanism with the lonely  mechanism, it is useful to define the following shorthands:
\begin{align*}
	OPT=\cost(A_{opt}(w'),w')\ , && OPT_{-i}=\cost(A_{opt}(\bot, w'_{-i}),w'_{-i})\ , \\ 
	LOPT=\cost(A_{lon}(w'),w')\ , && LOPT_{-i}=\cost(A_{lon}(\bot, w'_{-i}),w'_{-i}) \ ,
\end{align*}
where $A_{opt}$ is the optimal algorithm, and $A_{lon}$ is the lonely  algorithm, that is, the one computing the optimal solution using a single agent. Then the VCG payments \eqref{eq:VCG-payment} with the Clarke pivot rule in \eqref{eq:VCG:VP} can be rewritten as follows in the two cases:
\begin{align}
P_i(w') &= OPT_{-i} - \left(OPT - w_i' \cdot d_i(A_{opt}(w'))\right)\ , \label{eq:pay:VCG-opt} \\
P_i(w') &= LOPT_{-i} - \left(LOPT - w_i' \cdot d_i(A_{lon}(w'))\right)\ . \label{eq:pay:VCG-lonley}
\end{align}
Below, we usually omit the agents' weights $w'=w$ whenever they are 
clear from the context.
\begin{fact}\label{fact:lonley-pay}
	The lonely mechanism pays the selected agent $i$ an amount  $LOPT_{-i}$, and the other agents get no payment. This is because $LOPT = w_i' \cdot d_i(A_{lon}(w'))$ when $i$ is selected, and for all non-selected agent $j\neq i$, $LOPT=LOPT_{-j}$ and $d_j(A_{lon(w')})=0$.
\end{fact}

\begin{theorem}
	\label{th:mechanism-design:total-payments}
	For a single package, there are instances where the optimal mechanism pays a total amount of money larger than what the lonely mechanism does. Moreover, there are instances in which the opposite happens, that is, the lonely mechanism pays more. 
\end{theorem}

\begin{proof}
	First we prove that there are instances, in which the optimal mechanism pays more.
	The example in Figure~\ref{fig:frugality} (left)
	shows an instance and its optimal solution. 
	The optimum for the instance in which agent $1$ is not present has cost $OPT_{-1}=40$ (let agent $3$ do the whole work). More generally, one
	can check that $OPT_{-1} = OPT_{-2} = 40$ and $OPT_{-3}=45$, and obviously $OPT = 9 + 12 + 16 = 37$.
	Using these values for the payments \eqref{eq:pay:VCG-opt}, 
	we get
	$P_1 = 40 - (37 - 9) = 12$, $P_2 = 40 - (37 - 12) = 15,\ P_3 = 45 - (37 - 16) = 24$,
	for a total amount of $51$ payed by this mechanism to the agents. The lonely mechanism will instead select agent $3$ and pay only
	this agent (see Fact~\ref{fact:lonley-pay}) an amount $LOPT_{-3}= 6\cdot (2+2+4) = 48$, which is the lonely optimum for the instance where agent $3$ is not present. 

	\begin{figure}[t!]
		\includegraphics[width=\linewidth]{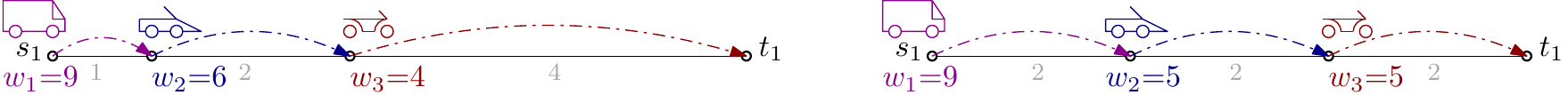}
		\caption{Delivery examples and their optimal solution.
		(left) The optimal mechanism pays more than the lonely mechanism.
		(right) The lonely mechanism pays more than the optimal mechanism.}
		\label{fig:frugality}
	\end{figure}

	On the other hand, the example in Figure~\ref{fig:frugality} (right) shows an instance in which the lonely mechanism pays more.
	The optimal solution has cost is $OPT=18+10+10=38$ and we get $OPT_{-1}=40$, $OPT_{-2}= 46$ and $OPT_{-3}=38$, resulting in the payments
	$P_1 = 40 - (38-18) = 20,\ P_2 = 46 - (38 - 10) = 18,\ P_3 = 38 - (38 - 10) = 10$ for a total amount of $48$.
	The lonely mechanism will instead select agent $2$ and pay $LOPT_{-2}= 50$.
\end{proof}

We remark that the payments given by \eqref{eq:VCG-payment} and \eqref{eq:VCG:VP} guarantee \emph{voluntary participation}, and therefore the mechanism \emph{must} pay a total amount of at least the optimum. We show that both
mechanisms pay only a small constant factor over the optimum, 
except when a single agent can do the work for a much cheaper price than the
others (as shown in Example~\ref{ex:monopoly}): 

\begin{definition}[monopoly-free]
	 We say that an instance with a single package is \emph{monopoly-free} if there is an optimal solution which uses at least two agents.
\end{definition}
\begin{example}\label{ex:monopoly}
	Both the exact and the lonely mechanism perform equally bad if, for instance, there is only one very cheap agent and another very expensive
	one. Consider two agents of weights $w_1=\epsilon$ small, and $w_2=L$ large, both agents sitting on the starting position $s_1$ of the package. 
	In this case the two mechanisms output the same solution and payment: 
	Agent $1$ does the whole work and gets an amount of money given by the best alternative solution. We have $P_1=$ 
	$w_2 \dist(s_1,t_1)=L \dist(s_1,t_1)$, while the optimum is $w_1 \dist(s_1,t_1)=\epsilon \dist(s_1,t_1)$.
\end{example}

\begin{theorem}
	In any single package monopoly-free instance, the optimal mechanism pays a total amount of money which is at most twice the optimum.
\end{theorem}

\begin{proof}
	The optimal solution selects a certain number $\ell\geq 2$ of agents and
	assigns to each of them some path. By renaming the agents, we can therefore assume that the optimum cost is of the form 
	$OPT = w_1d_1 + w_2 d_2 + \cdots + w_\ell d_\ell$, where no agent
	appears twice and the weights must satisfy
	$w_i \geq  w_{i+1}$ and $w_i \leq  2w_{i+1}$,
	for otherwise agent $i$ can replace agent $i+1$ or vice versa. 
	We shall prove below that
	\begin{equation}
		\label{eq:OPT-i}
		OPT_{-i} \leq OPT + w_i d_i. 
	\end{equation}
	Using VCG payments \eqref{eq:VCG-payment}, we obtain from \eqref{eq:OPT-i} that every agent is payed at most twice her cost, $P_i \leq  OPT + w_i d_i  -(OPT-  w_i d_i) = 2w_i d_i$,
	which then implies the theorem.	
	To complete the proof, we show  \eqref{eq:OPT-i} by distinguishing two cases. 
	For $i < \ell$,
	we can replace agent $i$ with agent $i+1$ who then has to travel an additional distance of at most $2d_i$ to reach $i$ and come back to its position.
	This gives an upper bound:
	$
	OPT_{-i} \leq  OPT - w_i d_i + w_{i+1}2d_i  \leq OPT + w_i d_i$, where the last inequality is due to $w_{i+1} \leq w_i$. 
	For $i=\ell$, we replace agent $i$ by agent $\ell -1$ who travels an extra amount $d_\ell$, and obtain this upper bound: 
	$
	OPT_{-\ell} \leq  OPT - w_\ell d_\ell + w_{\ell-1}d_\ell  \leq OPT + w_\ell d_\ell
	$, where the last inequality is due to $w_{\ell-1}\leq 2w_\ell$.
\end{proof}

\begin{theorem}
	In any single package monopoly-free instance, the lonely mechanism pays at most $2BoC$ times the optimum, where $BoC = 1/ \ln 2\approx 1.44$
	(by Theorem~\ref{thm:boc}).
\end{theorem}
\begin{proof}
	Let $i$ be the selected agent. A crude upper bound on the sum of the payments can be obtained via Fact~\ref{fact:lonley-pay}, 
	where the second inequality requires the instance to be monopoly-free:
	$P_i = LOPT_{-i} \leq \BoC \cdot OPT_{-i} \stackrel{\eqref{eq:OPT-i}}{\leq} \BoC \cdot (OPT + w_i d_i) \leq \BoC \cdot 2 \cdot OPT$.
\end{proof}

\section{Conclusion and open questions}

This work initiates the study of truthful mechanisms for delivery in a natural setting where mobile agents are \emph{selfish} and  can speculate about their own energy consumption rate.
We considered mechanisms which are \emph{truthful}, run in \emph{polynomial-time}, have a worst-case \emph{approximation} guarantee and, possibly, do not pay the agents abnormal amounts (\emph{frugality}).
We provide such polynomial-time truthful approximation mechanisms for two cases: 
(1) when the consumption rate of different agents are similar, i.e., the ratio $\max\frac{w_i}{w_j}$ is bounded (see Theorem~\ref{th:mechanism-general}) and 
(2) when the number of packages $m$ is small (see Theorem~\ref{thm:mechanism-constantm-improved}).
For a single package, we also gave bounds on the frugality of two natural truthful mechanisms which return the optimum or an approximation of the optimum. 

The main open question is whether there exists a polynomial-time constant-factor approximation independent of the weights.
This remains an intriguing open question even for a constant number of agents $k$, for which we presented an exponential-time truthful approximation mechanism.

\subparagraph{Acknowledgments.}
We would like to thank Jérémie Chalopin, Shantanu Das, Yann Disser, Jan Hackfeld, and Peter Widmayer for some insightful discussions.
We are very grateful for the feedback provided by the anonymous reviewers.
This work was partially supported by the SNF (Project 200021L\_156620, Algorithm Design for Microrobots with Energy Constraints).

\bibliographystyle{abbrv}
\bibliography{bib-mechanismagents}

\end{document}